\documentclass[a4paper,12pt]{article}

\usepackage{amsmath,amsthm,amssymb,latexsym,stmaryrd}
\usepackage{tikz}
\usepackage{tikz-cd}
\usepackage[utf8]{inputenc}

\title{ \textbf{The $\infty$-groupoid generated by an arbitrary topological $\lambda$-model}}
\author{
	Daniel O.\ Martínez-Rivillas, Ruy J.G.B.\  de Queiroz
}
\date{\small{\today}}

\begin{document}
\maketitle

\begin{abstract}
The lambda calculus is a universal programming language. It can represent the computable functions, and such offers a formal counterpart to the point of view of functions as rules. Terms represent functions and this allows for the application of a term/function to any other term/function, including itself. The calculus can be seen as a formal theory with certain pre-established axioms and inference rules, which can be interpreted by models. Dana Scott proposed the first non-trivial model of the extensional lambda calculus, known as $ D_\infty$, to represent the $\lambda$-terms as the typical functions of set theory, where it is not allowed to apply a function to itself. Here we propose a construction of an $\infty$-groupoid from any lambda model endowed with a topology. We apply this construction for the particular case $D_\infty$, and we see that the Scott topology does not provide enough information about the relationship between higher homotopies. This motivates a new line of research focused on the exploration of $\lambda$-models with the structure of a non-trivial $\infty$-groupoid to generalize the proofs of term conversion (e.g., $\beta$-equality, $\eta$-equality) to higher-proofs in $\lambda$-calculus.

\end{abstract}


\newtheorem{defin}{Definition}[section]
\newtheorem{teor}{Theorem}[section]
\newtheorem{corol}{Corollary}[section]
\newtheorem{prop}{Proposition}[section]
\newtheorem{rem}{Remark}[section]
\newtheorem{lem}{Lemma}[section]
\newtheorem{nota}{Notation}[section]
\newtheorem{ejem}{Example}[section]

\section{Introduction}

The lambda calculus is a programming language in which functions are seen as rules instead of sets. Since the origins of Computer Science, lambda calculus has been widely used as a formal counterpart to the notion of algorithm. Since it constitutes the essence of functional programming languages, for example, the ML family, such as CALM and SML, it is also of great interest the study of lambda calculus from the point of view of types, not least because the discipline of typing allows the detection of errors without the need to execute a given program. There are extensions of typed lambda calculus such as Martin-L\"of's Type Theory (MLTT), also known as \textit{Intuitionistic Type Theory}, where unlike lambda calculus allows for dependent types such as \textit{the identity type} $I_A(a,b)$, with $A$ being a type, and $a$ and $b$ being terms of type $A$.  

 \medskip Under the so-called Curry--Howard isomorphism, the type $I_A(a,b)$ corresponds to the proposition which says that $a$ is equal to $b$ in the type $A$, and its terms (if these exist) would be proofs of this equality. If there is a proof $p$ of $I_A(a,b)$, this does not imply that $a\equiv b$,  (extensionally equal, i.e., the reflexivity term $ref(a)$ is an inhabitant of $I_A(a,b)$), since it can happen that $a$ and $b$ are extensionally equal, but not necessarily intensionally equal. In the other direction, if $a \equiv b$ then this implies that there is a proof of equality of $I_A(a,b)$, e.g.,  the proofs $ref(a)$ and $ref(b)$. Thus, $I_A(a,b)$ is a weaker type of equality than extensional equality, but it can gather more information regarding the multiple proofs of equality. The type $I_A(a,b)$ is known as a \textit{propositional equality}. 

\medskip Furthermore, given two proofs of equality $p$ and $q$ in $I_A(a,b)$, we can consider the type $I_{I_A(a,b)}(p,q)$ as the type of terms representing proofs that $p$ is equal to $q$, and thus continue iterating indefinitely to obtain an infinite sequence of higher identity types, which carries an algebraic structure known as $\infty$-groupoid. 

\medskip In Homotopy Type Theory (HoTT), the types of MLTT  are interpreted as topological spaces, and proofs of identity $p$ of $I_A(a,b)$ are seen as continuous paths from $a$ to $b$. The proofs $h$ of identity proofs in $I_{I_A(a,b)}(p,q)$  are interpreted as  homotopies $h$ from $p$ to $q$, and so on, the fundamental $\infty$-groupoid $\Pi_\infty A(a,b)$ is obtained this way.

\medskip Since MLTT is a formalization of typed lambda calculus, one is allowed to suspect that it can also carry an algebraic $\infty$-groupoid structure. If the types are interpreted, not as sets, but as topological spaces, one has a rich mathematical structure to model complex phenomena such as computations. Our motivation is to study type-free lambda calculus from a model that allows us to see the $\infty$-groupoid structure.

 \medskip The system of MLTT with identity types (Martin-Löf, 1975), was developed originally to give a formalization of constructive mathematics in which there would be formal counterparts to proofs of identity statements. By studying the relationship between two proofs of a proposition, the relation of two higher-proofs of a proposition on proofs, and so on, one could have the formal counterpart to such a hierarchical structure of globular sets. Later (Hofmann; Streicher, 1994) comes up with the idea of using higher-order categories for the interpretation of MLTT, and later on (Awodey; Warren, 2009) manages to establish the connection between MLTT and algebraic topology, in the sense of which types of identity can be interpreted as an equivalence of homotopies. This led (Voevodsky, 2010) to formulate the Univalence Axiom, which gives rise to Homotopy Type Theory (HoTT) (Program, 2013), where MLTT has a structure of $\infty$-groupoid (Berg; Garner, 2011). Also, Voevodsky proved that HoTT has a model in the category of Kan complexes ($\infty$-groupoids) (Kapulkin; Lumsdaine; Voevodsky, 2012) for Univalence Axiom and (Lumsdaine; Shulman, 2020) for higher inductive types. 

\medskip Since MLTT is an extension of simply-typed $\lambda$-calculus, it can also be seen, as an $\infty $-groupoid in this topological interpretation given by (Scott, 1993). Also, for type-free $\lambda$-calculus, Dana Scott presented the model $D_\infty$ to gives an interpretation to $\lambda$-terms into the theory of ordered sets. This model is semantically rich in the sense that it is an ordered set with a topology. It allows for generating a $\lambda$-model $\mathfrak{D}_\infty$,  with a structure of an $\infty$-groupoid, and an operation of composition between cells. 

\medskip In this work, we build an $\infty$-groupoid $\mathfrak{D}$ from a topological space $D$ through higher groups of homotopy (Greenberg, 1967; Hatcher, 2001).  We show how to calculate all higher groups generated by any c.p.o.,  with the Scott topology (Acosta; Rubio, 2002).  We apply the construction of Section 3 for the particular case of the c.p.o.\ $D_\infty$ to obtain an $\infty$-groupoid $\mathfrak{D}_\infty$, and prove that this is isomorphic to $D_\infty$, which shows that $\mathfrak{D}_\infty$  is indeed an extensional $\lambda$-model. Unfortunately, the $\infty$-groupoid associated with $D_\infty$ and $\mathfrak{D}_\infty$ turns out to be trivial. So in Section 5, we explain with geometric intuition the purpose of the search for higher non-trivial $\lambda$-models, which we called \textit{homotopic $\lambda$-models}, and which are studied in more detail by (Martínez; de Queiroz, 2020) in light of simplicial sets and Kan complexes (Goerss; Jardine, 2009).

\section{Preliminaries}

In this section we present some basic notions about $\infty$-groupoids, topological spaces, continuity, higher fundamental groups and extensional lambda models, to set up the groundwork for this paper. All the proofs of results can be found in the suggested references.

\subsection{Identity types in HoTT}

The Homotopy Types Theory (HoTT) corresponds to the axioms and rules of the intensional version of Intuitionistic Type Theory (ITT) plus the univalence axiom and higher inductive types. It was created to give a new foundation of mathematics and facilitate the translation of mathematical proofs into computer programs. In this way, it allows computers to verify mathematical proofs with high deductive complexity. 

\medskip HoTT facilitates the understanding of ITT by allowing for an interpretation based on the geometric intuition of Homotopy Theory. For example, a type $ A $ is interpreted as the topological space, a term $ a: A $ as the point $ a \in A $, a dependent type $ x: A \vdash B (x) $ as the fibration $ B \rightarrow A $, the identity type $I_A$ as the space path $A^I$,  a term $p:I_A(a,b)$    as the path $p: a\rightarrow b$,  the term $\alpha:I_{I_A(a,b)}(p,q)$ as the homotopy $\alpha:p\Rightarrow q$ and so on. 

\medskip Among the dependent types arise the identity types, which were inductively defined by Martin-Löf analogously to the inductive definition of natural numbers, according to the axioms.

\begin{enumerate}
	\item[A1.] If $A$ is a type, and $a$ and $b$ are terms that inhabit it, writing $a,b:A$, there is an identity type denoted by $I_A(a,b)$ (or $a=_Ab$),
	\item [A2.] if $A$ is a type and $a:A$, there is a term $ref(a):a=_Aa$ (reflexivity),
	\item[A3.] if $A$ is a type, $a:A$ and $P(b,e)$ is a family of types depending on parameters $b:A$ and $e:I_A(a,b)$. In order to define any term $f(b,e):P(b,e)$,  it suffices to provide a  term $p:P(a,ref(a))$. The resulting term $f$ may be regarded as having been completely
	defined by the single definition  $f(a,ref(a)):=p$.
\end{enumerate}

Here the axiom A3 is analogous to induction axiom of natural numbers and by its way the properties 
\begin{enumerate}
	\item given the type $a=_Ab$, there is the no-void type $b=_Aa$ (symmetry),
	\item let the types $a=_Ab$ and $b=_Ac$, there is the no-void type $a=_Ac$ (transitivity),  
\end{enumerate}
can be proved by induction, see (Martin-Löf; 1973) and (Grayson; 2018).

\medskip In HoTT,  the property 1 is proved simply by inverting any path $p:a=_Ab$, i.e., $p^{-1}:b=_Aa$. The property 2 can be proved by concatenating any paths $p:a=_Ab$ and $q:b=_Ac$, this is, $p\ast q:a=_Ac$. The term $ref(a)$ from axiom A2, is interpreted as the constant path on the point $a$, denoted by $c(a)$ or also writing as $1_a$. The axiom A3 can be seen as, given a path  $e:a=_Ab$  and a proof of any property $P(a)$, this proof can be transported by way of the path $e$ to give a proof of the property $P(b)$. For more information on HoTT see (Program, 2013). 

\subsection{Strict $\infty$-groupoids}

In the literature we can find two types of $\infty$-groupoids: strict and  weak $\infty$-groupoids (Leinster, 2003). For this paper we shall only work with the former. It is well known that every strict $\infty$-groupoid is weak. It is usual to call a  weak $\infty$-groupoid just $\infty$-groupoid.

\begin{defin}[\textbf{$\infty$-globular set}]
An $\infty$-globular set $D$ is a diagram
	$$\cdots\rightrightarrows^s_t D_n\rightrightarrows^s_t D_{n-1}\rightrightarrows^s_t\cdots\rightrightarrows^s_t D_1\rightrightarrows^s_t D_0,$$
	of sets and functions such that
	$$s(s(d))=s(t(d)), \,\,\,\,\, t(s(d))=t(t(d)),$$
	for all $n\geq 2$ and $d\in D_n$. 
\end{defin}

\begin{defin}
 Let $D$ be a globular set and $n\in\mathbb{N}$. For each $0\leq p<n$ define the relation into $D_n\times D_n$ as the set
 $$D_n\times_{D_p} D_n=\{(d',d)\in D_n\times D_n:t^{n-p}(d)=s^{n-p}(d')\}.$$
\end{defin}

\begin{defin}[\textbf{strict $\infty$-groupoid}]
	Let $n$ be a natural number such that $n>0$. A strict $\infty$-groupoid is an $\infty$-globular set $D$ equipped with
	\begin{itemize}
		\item a function $\circ_p:D_n\times_{D_p}D_n\rightarrow D_n$ for each $0\leq p<n$, where $\circ_p(d',d):=d'\circ d$ and call it a composite of $d'$ and $d$,
		
		\item a function $i:D_n\rightarrow D_{n+1}$ for each $n\geq 0$, where $i(d):=1_d$ and call it the identity on $d$,
	\end{itemize} 
satisfying the following axioms:

\begin{enumerate}
	\item[a.]$($sources and targets of composites$)$ if $0\leq p<n$ and $(d',d)\in D_n\times_{D_p} D_n$ then
	\begin{align*}
	& s(d'\circ_p d)=s(d) \hspace{1.7cm} and\hspace{0.5cm}t(d'\circ_p d)=t(d')\hspace{1.4cm} if \,\, p=n-1, \\
	& s(d'\circ_p d)=s(d')\circ_p s(d)\hspace{0.5cm} and\hspace{0.5cm}t(d'\circ d)=t(d')\circ_p t(d)\hspace{0.5cm} if \,\, p\leq n-2,
	\end{align*}
	
	\item[b.] $($sources and targets of identities$)$ if $0\leq p<n$ and $d\in D_n$ then $s(1_d)=d=t(1_d)$,
	
	\item[c.]$($associativity$)$ if $0\leq p<n$ and $d,d',d''\in D_n$ with $(d'',d'),(d',d)\in D_n\times_{D_p} D_n$  then
	$$(d''\circ_p d')\circ_p d=d''\circ_p(d'\circ_p d),$$ 
	
	\item[d.]$($identities$)$ if $0\leq p<n$ and $d\in D_n$ then 
	 $$i^{n-p}(t^{n-p}(d))\circ_p d=d=d\circ_p i^{n-p}(s^{n-p}(d)),$$ 
	 
	 \item[e.] $($binary interchange$)$ if $0\leq q<p<n$ and $d,d',e,e'\in D_n\times_{D_q} D_n$ with 
	 $$(e',e),(d',d)\in D_n\times_{D_p} D_n,\hspace{0.3cm}(e',d'),(e,d)\in D_n\times_{D_q} D_n,$$
	 then
	 $$(e'\circ_p e)\circ_q (d'\circ_p d)=(e'\circ_q d')\circ_p(e\circ_q d),$$
	 
	 \item[f.] $($nullary interchange$)$ if $0\leq q<p<n$ and $d,d'\in D_p\times_{D_q} D_p,$ then
	 $1_{d'}\circ_q1_d=1_{d'\circ_q d}$.  
	 
	 \item[g.] $($inverse$)$ if $0\leq p<n$ and $d\in D_n$ then exist $\bar{d}\in D_n$ with $s(\bar{d})=t(d)$, $t(\bar{d})=s(d)$ such that
	 $$\bar{d}\circ_p d=i^{n-p}(s^{n-p}(d)),\hspace{0.3cm} d\circ_p\bar{d}=i^{n-p}(t^{n-p}(d)). $$
\end{enumerate}
\end{defin}

If $d\in D_n$, we say that $d$ is an $n$-cell or an $n$-isomorphism from some $a\in D_{n-1}$ to some $b\in D_{n-1} $. In this case $a=s(d)$ and $b=t(d)$. Or, in other words, we say that $a$ and $b$ are $n$-equivalent if there is an $n$-isomorphism between $a$ and $b$.

\medskip For example, in the fundamental $\infty$-groupoid $\Pi_\infty (D)$, where $D$ is a topological space and $D_n=\Pi_{n}(D)$, the $n$-isomorphisms are the $n$-paths class $[p]$ in $\Pi_n (D)$. Even though the $n$-path $p:a\rightsquigarrow b$  may not satisfy the properties of a strict $\infty$-groupoid, the $n$-paths class $[p]$ does satisfy them. Thus $\Pi_\infty (D)$ is not a strict $\infty$-groupoid, but it is a weak $\infty$-groupoid. This subject will be addressed in more details in Section 4. 

\begin{nota}
	Write $(a\simeq_n b)$ for the set of all $n$-isomorphisms between $a$ and $b$. Note that this set can be empty.
\end{nota}

Since $\simeq_n$ is an equivalence relation and the set $(a\simeq_n b)$ can have cardinality greater than one, we can see $\simeq_n$ as an intensional equality $=_n$, i.e., the equivalence $a\simeq_n b$ can be seen as the intensional equality $a=_{h}b$, which motivates a more precise definition of intensional equality between $n$-cells.  

\begin{defin}[Extensional and intensional equality]
	Two $n$-morphisms $a$ and $b$ are intentionally equal, $a=_{h}b$, if there is $d:(a\simeq_n b)$. Two $n$-morphisms $a$ and $b$ are extensionally equal, $a=b$, if the identity $n$-morphism  $1_a:(a\simeq_n b)$.    
\end{defin}

Note that if $a=b$ then $a=_{h}b$, since $1_{a}:(a\simeq_n b)$. The converse does not always hold. For example, for the fundamental groupoid $\Pi_1\{0,1\}$, where $\{0,1\}$ has the topology $\{\{0,1\},\emptyset\}$, we have that $0=_{h}1$, but $0\neq 1$. 

\subsection{Higher groups of homotopy}

Next we define closed $n$-paths on $d_0\in D$, with the purpose of building the fundamental $n$-group $\pi_n (D,d_0)$. For more details, see (Greenberg, 1967), (Hatcher, 2001) and (May, 1999).

\begin{defin}[Border of a set]
	Let $D$ be a topological space. Define the border of a set $X\subseteq D$, denoted by $\partial X$, as the set of points $x\in D$ such that for each neighborhood $V$ of $x$, $V\cap X\neq\emptyset$ and $V\cap (D-X)\neq\emptyset$. 
\end{defin}

\begin{defin}[Closed $n$-path]\label{definition-path-closed}
	Let $D$ be a topological space. A closed $n$-path on $d_0\in D$ is a continuous map $\sigma:[0,1]^n\rightarrow D$ which sends border of $[0,1]^n$ to $d_0$.
	
	\medskip It defines the product of closed $n$-paths, $\alpha\ast\beta=\gamma$, as the closed $n$-path
		\begin{equation*}
	\gamma(t_1,\ldots,t_n)=
	\begin{cases}
	\alpha(2t_1,t_2\ldots,t_n) & \text{if \, $0\leq t_1\leq \frac{1}{2}$,}\\
	\beta(2t_1-1,t_2\ldots,t_n) & \text{if \, $\frac{1}{2}\leq t_1\leq 1$.}
	\end{cases}
	\end{equation*}
\end{defin}

This product is known in the literature as concatenation operator between the $n$-paths $\alpha$ and $\beta$.

\begin{defin}[Homotopic $n$-paths]
	Two $n$-paths $\alpha$, $\beta$ are homotopic in $d_0$, $\alpha\simeq\beta$, if there exists a continuous map $H:[0,1]\times[0,1]^n\rightarrow D$, such that 
	\begin{enumerate}
		\item[a.] $H(0;t_1,\ldots,t_n)=\alpha(t_1,\ldots,t_n)$, for $(t_1,\ldots,t_n)\in[0,1]^n$,
		\item[b.] $H(1;t_1,\ldots,t_n)=\beta(t_1,\ldots,t_n)$, for $(t_1,\ldots,t_n)\in[0,1]^n$,
		\item[c.] $H(s;t_1,\ldots,t_n)=d_0$, for each $s\in[0,1]$ and each $(t_1,\ldots,t_n)\in\partial([0,1]^n)$.
	\end{enumerate} 
\end{defin}

This homotopy is an equivalence relation, and one writes $[\sigma]$ for the class of homotopic $n$-paths to the $n$-path $\sigma$. The set of all homotopy classes is denoted by $\pi_n(D,d_0)$, where the product between classes is defined in the natural way $[\alpha]\ast[\beta]:=[\alpha\ast\beta]$.

\medskip While in Definition~\ref{definition-path-closed} of closed $n$-path $p$ based in a point $a\in D$, the image of the $[0,1]^n$ border fell on $a$, $p(\partial([0,1]^n))=\{a\}$, for $n$-paths the image of the lower border falls at some point $a\in D$, i.e., $p(t_1,\ldots,t_{n-1},0)=a$, and the upper image falls at some point $b\in D$, i.e., $p(t_1,\ldots,t_{n-1},1)=b$.

\begin{teor}
	$\pi_n(D,d_0)$ is a group, called the fundamental group of $D$ on $d_0$ of dimension $n$.
\end{teor}

The identity element of the group $\pi_n(D,d_0)$ is the homotopy class of the constant $n$-path $c^n(d_0):[0,1]^n\rightarrow D$, i.e., $c^n(d_0)(t_1,\dots,t_n)=d_0$ for all $t_1,\ldots,t_n$ in $[0,1]$.

\begin{defin}[Homotopic functions]
	Let $D$ be a topological space. Two continuous functions $f,g:D\rightarrow D$ are homotopic, $f\simeq g$, if there exists a continuous map $H:D\times [0,1]\rightarrow D$ such that $H(d,0)=f(d)$ and $H(d,1)=g(d)$ for all $d\in D$. $H$ is called a homotopy between continuous functions.
 \end{defin}

\begin{defin}[Contractible space]
	A topological space $D$ is said to be contractible if there exists a constant function $f_{c}:D\rightarrow D$, $f_{c}(d)=c$ for each $d\in D$, homotopic to the identity function $I_D:D\rightarrow D$. The homotopy $H:f_c\simeq I_d$ is called a contraction.
\end{defin}

\begin{teor}\label{trivialgroupcontrac-theorem}
	If $D$ is contractible, then $\pi_n(D,d_0)=\{[c^n(d_0)]\}$ for each $n\geq 0$.
\end{teor}

\begin{rem}
	Let $D$ contractible. For each $t_1,\dots,t_n\in[0,1]$ we have
$$c^n(d_0)(t_1,t_2,\ldots,t_n)=e_n(t_1)(t_2)\cdots(t_n)=d_0,$$ 
where $e_n:[0,1]\rightarrow \Pi_{n-1}(D,d_0)$ is the constant path in $e_{n-1}$ defined by recursion 
\begin{align*}
\hspace{1cm}& e_0=d_0, \\
& e_n=c_{e_{n-1}},
\end{align*}
then $\pi_n(D,d)\cong \{e_n\}$ $($group isomorphisms$)$. 
\end{rem}

\subsection{The $\boldsymbol{\lambda}$-model Scott's $\boldsymbol{D_\infty}$}

Next we present a brief introduction to the theories and extensional lambda models. For the proofs of the theorems, one is referred to (Hyndley; Seldin, 2008).

\begin{defin}[Directed set]
	Let $(D,\sqsubseteq)$ be a partial ordering. A non-empty subset $X\subset D$ is said to be directed if for all $a,b\in X$, there exists $c\in X$ such that $a\sqsubseteq c$ and $b\sqsubseteq c$.
\end{defin}

\begin{defin}[Complete partial orders, c.p.o.'s]
	A c.p.o.\ is a partially ordered set $(D,\sqsubseteq)$ such that
	\begin{enumerate}
		\item[a.] $D$ has a least element $($denoted $\perp$$)$,
		\item[b.] every directed subset $X\subset D$ has least upper bound, l.u.b., $($denoted $\bigsqcup X)$.
	\end{enumerate}
The pair $(D,\sqsubseteq)$ will be denoted $D$.
\end{defin}

\begin{defin}[The set $\mathbb{N}^+$]
	Choose any object $\perp\notin\mathbb{N}$, and define $\mathbb{N}^+=\mathbb{N}\cup\{\perp\}$. For all $a,b\in\mathbb{N}^+$, define
	$$a\sqsubseteq b\Longleftrightarrow (a=\perp\,\, and \,\, b\in\mathbb{N}) \,\, or \,\, a=b.$$
\end{defin}

Clearly $\mathbb{N}^+$ is a c.p.o., since every directed subset is finite so has l.u.b.\ and $\perp$ is the least element as seen in Figure \ref{fig03}. 

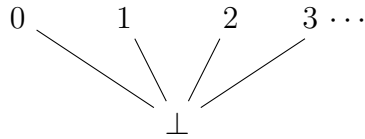
\begin{figure}[!htb]
	\centering
	\begin{tikzpicture}[scale=.7]
	\node (a) at (-3,0) {$0$};
	\node (b) at (-1,0) {$1$};
	\node (c) at (1,0) {$2$};
	\node (d) at (3,0) {$3\,\cdots$};
	\node (zero) at (0,-2) {$\bot$};
	\draw (zero) -- (a)  (b) -- (zero) -- (c)  (d) -- (zero);
	\end{tikzpicture}
	\caption{The c.p.o. $\mathbb{N}^+$}\label{fig03}
\end{figure}
Every c.p.o.\ has a topology called the Scott topology, which we define below.

\begin{defin}[Final and inaccessible set]
	Let $D$ a c.p.o. and $A\subseteq D$. The set $A$ is final if it satisfies
	$$a\in A \,\, and \,\, a\sqsubseteq b\Longrightarrow b\in A,$$
	and $A$ is inaccessible by directedness if for every directed subset $X$ of $D$,
	$$\bigsqcup X\in A\Longrightarrow X\cap A\neq \emptyset.$$
\end{defin}

\begin{defin}[Scott topology]
	Let $D$ be a c.p.o. The Scott topology is defined as follows 
	\begin{center}
		$\sigma=\{A\subseteq D: A$ is final and inaccessible by directedness$\}\cup\{\emptyset\}.$ 
	\end{center} 
\end{defin}

\begin{defin}[Continuous function]
	Let $D$ and $D'$ be c.p.o.'s. A function $f:D\rightarrow D'$ is continuous if for each directed subset $X\subset D$, $f(\bigsqcup X)=\bigsqcup f(X)$. 
\end{defin}

\begin{defin}
	For c.p.o.'s $D$ and $D'$, define $[D\rightarrow D']$ to be set of all continuous functions from $D$ to $D'$ in Scott's topology. For $\phi,\psi\in[D\rightarrow D']$, define
	$$\phi\sqsubseteq\psi\Longleftrightarrow (\forall d\in D)(\phi(d)\sqsubseteq'\psi(d)).$$
\end{defin}

\begin{rem}\label{functionsetcpo-lemma}
	There is a result which says that if $D$ and $D'$ are c.p.o.'s, then $[D\rightarrow D']$ is a c.p.o. (see proof in (Hindley; Seldin, 2008)). It would allow from a c.p.o.\ initially generate inductively an infinite sequence of c.p.o.'s as below in the Definition \ref{sequence-cpos-definition}. 
\end{rem}

\begin{defin}[Projections]\label{testenv-definition}
	Let $D$ and $D'$ be c.p.o.'s. A projection from $D'$ to $D$ is a pair $\left\langle \phi,\psi\right\rangle $ of functions, where $\phi\in[D\rightarrow D']$ and $\psi\in[D'\rightarrow D]$, such that 
	$$\psi\circ\phi=I_D, \hspace{0.3cm}\phi\circ\psi\sqsubseteq I_{D'}.$$
\end{defin}

\begin{defin}[Projection from $D_{n+1}$ to $D_{n}$]
	For every $n\geq 0$ define the projection $\langle\phi_n,\psi_n\rangle$ from $D_{n+1}$ to $D_n$ by the recursion
	\begin{align*}
	&\phi_0(d):=\boldsymbol{\lambda} a\in D_0.d, \hspace{3.5cm}\psi_0(g):=g(\perp_0), \\
	&\phi_{n+1}(d):=\phi_n\circ d\circ\psi_n, \hspace{1.8cm}\psi_{n+1}(g):=\psi_n\circ g\circ\phi_n,
	\end{align*}
	where $\boldsymbol{\lambda} a\in D_0.d\in[D_0\rightarrow D_0]$ is the constant function to $d\in D_0$.  
\end{defin}

\begin{defin}[The sequence $D_0,D_1,D_2,\ldots$]\label{sequence-cpos-definition}
	For each $n\geq 0$, define $D_n$ by recursion 
	\begin{align*}
	&\hspace{1.3cm}D_0:=\mathbb{N}^+, \\
	&\hspace{1.3cm}D_{n+1}:=[D_n\rightarrow D_n]. 
	\end{align*}
	The $\sqsubseteq$-relation on $D_n$ will be denoted just `$\,\,\sqsubseteq$'. The least element of $D_n$ will be denoted $\perp_n$.
\end{defin}

By Remark \ref{functionsetcpo-lemma}, every $D_n$ is a c.p.o.

\begin{defin}[Construction of $D_\infty$]\label{D_infty-definition}
	We define the c.p.o.\ $D_\infty$ to be the set of all infinite sequences
	$$d=\left\langle d_0,d_1,d_2,\ldots\right\rangle,$$
	such that $d_n\in D_n$ and $\psi_n(d_{n+1})=d_n$, for all $n\geq 0$, where $\psi_n$ is part of projection $\left\langle \phi_n,\psi_n\right\rangle $ from $D_{n+1}$ to $D_n$  (Hindley; Seldin, 2008).
	
	\medskip A relation $\sqsubseteq$ on $D_\infty$ is defined by 
	$$d\sqsubseteq d'\Longleftrightarrow (\forall n\geq 0)(d_n\sqsubseteq d_n').$$
\end{defin}

\medskip In (Barendregt, 1984) one has a way to prove that $D_\infty$ is a  $\lambda$-model through the result: If $D$ is a c.p.o.\ for which there exists a projection $\langle F, G\rangle$ from $[D\rightarrow D]$ to $D$ such that $G\circ F=I_{[D\rightarrow D]}$, then the triple  $\langle D,\bullet,\llbracket \, \rrbracket\rangle$ defined for every assignment $\rho:Var\rightarrow D$ by

\begin{enumerate}
	\item [(a)] $a\bullet b:=F(a)(b)$ for each $a,b\in D$,	
	\item [(b)] $\llbracket x \rrbracket_\rho:=\rho(x),$
	\item [(c)]$\llbracket PQ \rrbracket_{\rho}:=\llbracket P\rrbracket_{\rho}\bullet\llbracket Q\rrbracket_{\rho},$
	\item [(d)]$\llbracket\lambda x.P\rrbracket_{\rho}:=G(\boldsymbol{\lambda} d\in D.\llbracket P \rrbracket_{[d/x]_\rho})$, 
\end{enumerate}

\noindent is a $\lambda$-model. Also if $G\circ F=I_D$, i.e., $G=F^{-1}$ and $D\cong [D\rightarrow D]$, then  $\langle D,\bullet,\llbracket \, \rrbracket\rangle$ is an extensional $\lambda$-model.

\medskip So in the particular case of $D_\infty$, it holds that $D_\infty\cong [D_\infty\rightarrow D_\infty]$ where the isomorphism between c.p.o.'s $F:D_\infty\rightarrow[D_\infty\rightarrow D_\infty]$ is given for each $a\in D_\infty$ by
$$F(a)=\boldsymbol{\lambda} b\in D_\infty.a\bullet b,$$
whose inverse $F^{-1}:[D_\infty\rightarrow D_\infty]\rightarrow D_\infty$ corresponds to
$$F^{-1}=\boldsymbol{\lambda} f\in[D_\infty\rightarrow D_\infty].\bigsqcup_{n\geq 0}\phi_{n,\infty}(\boldsymbol{\lambda} a\in D_n.(\phi_{\infty,n}\circ f\circ\phi_{n,\infty})(a)).$$

Therefore $\langle D_\infty,\bullet,\llbracket \, \rrbracket\rangle$ is an extensional $\lambda$-model.

\section{The $\infty$-groupoid $\mathfrak{D}$ generated by an arbitrary topological space $D$} 

Next we show the construction of the $\infty$-groupoid from any topological space, through the use of higher fundamental groups. 

\begin{defin}[]\label{inftyglobular_D-definition}
	Let $D$ a topological space. Define the  $\infty$-globular set $\mathfrak{D}$ as the diagram
	
	$$\cdots\rightrightarrows^s_t \mathfrak{D}_n\rightrightarrows^s_t \mathfrak{D}_{n-1}\rightrightarrows^s_t\cdots\rightrightarrows^s_t \mathfrak{D}_1\rightrightarrows^s_t \mathfrak{D}_0,$$
	as follows 
	$$\mathfrak{D}_n:=\{\pi_n(D,d):d\in D\},$$
	where  $\pi_0(D,d):=d$, $\pi_n(D,d)$ is the fundamental group of dimension $n\geq 0$ and 
	$$s(\pi_{n+1}(D,d))=t(\pi_{n+1}(D,d)):=\pi_n(D,d).$$
\end{defin}

\begin{rem}
\label{testenv-remark}
Clearly $\mathfrak{D}$ is an $\infty$-globular set, since $s=t$, i.e., every morphism is an automorphism, then $s\circ s=s\circ t$ and $t\circ t=t\circ s$, see Figure \ref{fig:auto}. 
\end{rem}

\begin{figure}[ht!]
	\centering

	\fcolorbox{gray}{white}{\includegraphics[width=0.45\textwidth]{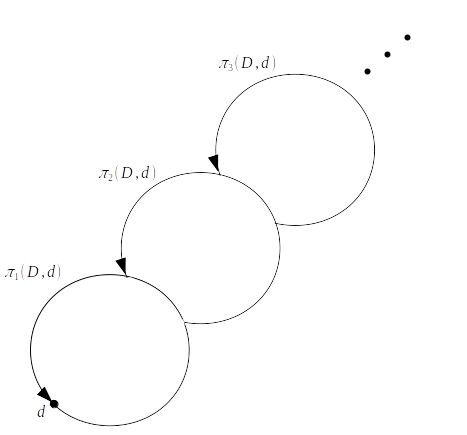}}
	\caption{$\infty$-Globular set $\mathfrak{D}$.}
	\label{fig:auto}
\end{figure}

\begin{rem}
\label{equalhgroup-remark}
For each $n\in\mathbb{N}$, the following holds:
$$\pi_n (D,d)=\pi_n (D,d')\Longleftrightarrow d=d'.$$
By the hypothesis that $D$ is connected by paths, and thus $\pi_n (D,d)\cong\pi_n (D,d')$ $($isomorphic groups$)$, it does not imply that they are equal.
\end{rem}

\begin{nota}
	For each $n\in\mathbb{N}$ and $d_0\in D$,  write
$$\mathfrak{D}_{n+1}(d_0):=\{\mathfrak{d}\in\mathfrak{D}_{n+1}:s(\mathfrak{d})=t(\mathfrak{d})=\pi_{n}(D,d_0)\}.$$
\end{nota}

\begin{prop}
\label{prop29}
For each $n\in\mathbb{N}$, the following holds
$$\mathfrak{D}_{n+1}(d_0)=\{\pi_{n+1}(D,d_0)\}.$$
\end{prop}

\begin{proof}
	Let $\mathfrak{d}\in\mathfrak{D}_{n+1}(d_0)$, then there is $d\in D$ such that $\mathfrak{d}=\pi_{n+1}(D,d)$. By Definition \ref{inftyglobular_D-definition} $s(\mathfrak{d})=t(\mathfrak{d})=\pi_{n}(D,d)$. Since $\mathfrak{d}\in\mathfrak{D}_{n+1}(d_0)$, then $\pi_{n}(D,d)=\pi_{n}(D,d_0)$, so $d=d_0$ by Remark \ref{equalhgroup-remark}, thus $\mathfrak{d}=\pi_{n+1}(D,d_0)$.   
\end{proof}

\begin{defin}[Diagonal]
	Let $D$ be a set. Define the diagonal on $D\times D$ as 
	$$Diag(D\times D):=\{(d,d')\in D\times D:d=d'\}$$
\end{defin}

\begin{lem}
\label{Diag-lemma}
For any natural number $n\geq 1$ and for each $0\leq p<n$,  $$\mathfrak{D}_n\times_{\mathfrak{D}_p}\mathfrak{D}_n=Diag(\mathfrak{D}_n\times \mathfrak{D}_n).$$
\end{lem}
\begin{proof}
	\begin{align*}
\hspace{1cm}\mathfrak{D}_n\times_{\mathfrak{D}_p}\mathfrak{D}_n&=\{(\mathfrak{d},\mathfrak{d}')\in \mathfrak{D}_n\times \mathfrak{D}_n:t^{n-p}(\mathfrak{d})=s^{n-p}(\mathfrak{d}')\} \\
& =\{(\mathfrak{d},\mathfrak{d}')\in \mathfrak{D}_n\times\mathfrak{D}_n:\pi_{n-(n-p)}(D,d)=\pi_{n-(n-p)}(D,d')\}  \\
& =\{(\mathfrak{d},\mathfrak{d}')\in \mathfrak{D}_n\times\mathfrak{D}_n:\pi_p(D,d)=\pi_p(D,d')\}  \\
& =\{(\mathfrak{d},\mathfrak{d}')\in \mathfrak{D}_n\times\mathfrak{D}_n:d=d'\}  \\
& =\{(\mathfrak{d},\mathfrak{d}')\in \mathfrak{D}_n\times\mathfrak{D}_n:\mathfrak{d}=\mathfrak{d}'\} \\
&=Diag(\mathfrak{D}_n\times \mathfrak{D}_n).
\end{align*}
\end{proof}
The Lemma \ref{Diag-lemma} indicates that it is enough to define the composition for pairs  $(\mathfrak{d},\mathfrak{d})\in\mathfrak{D}_n\times\mathfrak{D}_n $. 

\begin{defin}[Composition]\label{Composition-definition}
	For each $0\leq p<n$, define the composition of $\mathfrak{d}\in\mathfrak{D}_n$ with itself by
	$$\mathfrak{d}\circ_p\mathfrak{d}:=\{x \ast_p y: x,y\in\mathfrak{d}\}$$
	where $\ast_p$ is the path concatenation operator.  
\end{defin}

\begin{lem}
\label{idemcomposition-lemma}
For all $n\geq 1$, $0\leq p<n$ e $\mathfrak{d}\in\mathfrak{D}_n$, 
$$\mathfrak{d}\circ_p\mathfrak{d}=\mathfrak{d}.$$
\end{lem}
\begin{proof}
	Since $(\mathfrak{d},\ast_p)$ is a group, it is clear that $\mathfrak{d}\circ_p\mathfrak{d}\subseteq\mathfrak{d}$. On the other hand, if $x\in\mathfrak{d}$ then $x=x\ast_p e\in\mathfrak{d}\circ_p\mathfrak{d}$, where $e$ is the  identity element of the group $(\mathfrak{d},\ast_p)$. Thus $\mathfrak{d}\circ_p\mathfrak{d}=\mathfrak{d}$.  
\end{proof}

\begin{defin}[Identity]\label{Identity-definition}
	For each $n\in\mathbb{N}$ e $\mathfrak{d}=\pi_n(D,d)\in\mathfrak{D}_n$, define the identity function $i:\mathfrak{D}_n\rightarrow\mathfrak{D}_{n+1}$, $i(\mathfrak{d})=1_{\mathfrak{d}}$ as follows
	$$1_{\mathfrak{d}}:=\pi_{n+1}(D,d).$$
\end{defin}

\begin{teor}
	$\mathfrak{D}$ is an $\infty$-groupoid.
\end{teor}

\begin{proof}
	Let $n\geq 1$, $0\leq p<n$. For the axioms related to composition of morphisms, Lemma \ref{Diag-lemma} allows us to verify them by the composition of  $\mathfrak{d}=\pi_{n}(D,d)\in\mathfrak{D}_n$ with itself, then
	
	\medskip a. (sources and targets of composites)  by Lemma 3.2 and Definition 3.1 we have
	\begin{align*}
	& s(\mathfrak{d}\circ_p\mathfrak{d})=s(\mathfrak{d})=s(\mathfrak{d})\circ_p s(\mathfrak{d}), \\
	&  t(\mathfrak{d}\circ_p\mathfrak{d})=t(\mathfrak{d})=t(\mathfrak{d})\circ_p t(\mathfrak{d}),
	\end{align*}
	
	b. (sources and targets of identities)  by Definitions 3.1 and 3.3 
	\begin{align*}
	& s(1_\mathfrak{d})=s(\pi_{n+1}(D,d))=\pi_n(D,d)=\mathfrak{d}, \\
	& t(1_\mathfrak{d})=t(\pi_{n+1}(D,d))=\pi_n(D,d)=\mathfrak{d}, 
	\end{align*}
	
	c. (associativity) by Lemma 3.2
	$$(\mathfrak{d}\circ_p\mathfrak{d})\circ_p\mathfrak{d}=\mathfrak{d}\circ_p\mathfrak{d}=\mathfrak{d}\circ_p(\mathfrak{d}\circ_p\mathfrak{d}),$$
	
	d. (identities) by Definitions 3.1 , 3.3 and Lemma 3.2
	\begin{align*}
	&\hspace{1cm}i^{n-p}(t^{n-p}(\mathfrak{d}))\circ_p\mathfrak{d}=i^{n-p}(\pi_p(D,d))\circ_p\mathfrak{d}=\pi_{p+n-p}(D,d)\circ_p\mathfrak{d}=\mathfrak{d}\circ_p\mathfrak{d}=\mathfrak{d}, \\ 
	&\hspace{1cm}\mathfrak{d}\circ_p i^{n-p}(s^{n-p}(\mathfrak{d}))=\mathfrak{d}\circ_p i^{n-p}(\pi_p(D,d))=\mathfrak{d}\circ_p\pi_{p+n-p}(D,d)=\mathfrak{d}\circ_p\mathfrak{d}=\mathfrak{d},
	\end{align*} 
	
	e. (binary interchange) let $0\leq q<p<n$, by Lemma 3.2
	$$(\mathfrak{d}\circ_p\mathfrak{d})\circ_q(\mathfrak{d}\circ_p\mathfrak{d})=\mathfrak{d}\circ_q\mathfrak{d}=(\mathfrak{d}\circ_q\mathfrak{d})\circ_p(\mathfrak{d}\circ_q\mathfrak{d}),$$
	
	f. (nullary interchange) let $0\leq q<p<n$, by Lemma 3.2
	$$1_\mathfrak{d}\circ_q 1_\mathfrak{d}=1_\mathfrak{d}=1_{\mathfrak{d}\circ_q\mathfrak{d}},$$
	
	g. (inverse) by Lemma 3.2 and (d)
	$$\mathfrak{d}\circ_p\mathfrak{d}=\mathfrak{d}=i^{n-p}(t^{n-p}(\mathfrak{d}))=i^{n-p}(s^{n-p}(\mathfrak{d})),$$
thus $\mathfrak{d}$ is the inverse of itself.
\end{proof}

Next we define $\mathfrak{D}_\infty$ as a set in the sense of ZF set theory.   

\begin{defin}[The set $\mathfrak{D}_\infty$]\label{Setinftygrupoid_D-definition}
	Define $\mathfrak{D}_\infty$ as the set of all infinite sequences
	$$\mathfrak{d}:=\langle\mathfrak{d}_0,\mathfrak{d}_1,\mathfrak{d}_2,\dots\rangle,$$
	such that $\mathfrak{d}_n\in\mathfrak{D}_n$ (of the Definition \ref{inftyglobular_D-definition})  and $s(\mathfrak{d}_{n+1})=t(\mathfrak{d}_{n+1})=\mathfrak{d}_n$, for each $n\in\mathbb{N}$. 
\end{defin}

\begin{prop}\label{propo3.2}
\label{prop1}
$\mathfrak{d}\in\mathfrak{D}_\infty$ if and only if there exists $d\in D$ for all $n\in\mathbb{N}$, such that $\mathfrak{d}_n=\pi_{n}(D,d)$. 
\end{prop}
\begin{proof}
	Let $\mathfrak{d}\in\mathfrak{D}_\infty$, then $\mathfrak{d}_0\in\mathfrak{D}_0$, i.e., $\mathfrak{d}_0=\pi_0(D,d)$ for some and unique $d\in D$. Suppose that  $\mathfrak{d}_n=\pi_{n}(D,d)$ and we will prove by induction that $\mathfrak{d}_{n+1}=\pi_{n+1}(D,d)$. Since $\mathfrak{d}\in\mathfrak{D}_\infty$, by induction hypothesis $s(\mathfrak{d}_{n+1})=t(\mathfrak{d}_{n+1})=\mathfrak{d}_n=\pi_{n}(D,d)$. By Proposition \ref{prop29} we have $\mathfrak{d}_{n+1}=\pi_{n+1}(D,d)$. On the other hand, if $\mathfrak{d}_n=\pi_{n}(D,d)$ for every $n\in\mathbb{N}$, then $s(\mathfrak{d}_{n+1})=t(\mathfrak{d}_{n+1})=\pi_n(D,d)=\mathfrak{d}_n$ for all $n\in\mathbb{N}$, thus $\mathfrak{d}\in\mathfrak{D}_\infty$.     
\end{proof}

\begin{prop}
\label{prop37}
Let $\mathfrak{a}, \mathfrak{b}\in\mathfrak{D}_\infty$, such that $\mathfrak{a}_n=\pi_{n}(D,a)$ and $\mathfrak{b}_n=\pi_{n}(D,b)$ for all $n\in\mathbb{N}$, then  
$$\mathfrak{a}=\mathfrak{b}\Longleftrightarrow \mathfrak{a}_0=\mathfrak{b}_0.$$
\end{prop}
\begin{proof}
	If $\mathfrak{a}=\mathfrak{b}$, by Definition \ref{Setinftygrupoid_D-definition} we have $\mathfrak{a}_0=\mathfrak{b}_0$. If $a=\mathfrak{a}_0=\mathfrak{b}_0=b$, by the Proposition \ref{propo3.2} above $\mathfrak{a}_n=\pi_{n}(D,a)=\pi_{n}(D,b)=\mathfrak{b}_n$. 
\end{proof}

\section{Higher fundamental groups of a c.p.o.}

Next we show that every higher fundamental groupoid on any c.p.o.\ are trivial, particularly those generated by $D_\infty$.

\begin{lem}\label{bottomnotin-lemma}
	Let $D$ be a c.p.o.\ with the Scott topology. If $A\neq D$ is an open, then $\perp\notin A$. 
\end{lem}
\begin{proof}
	Let $A$ be an open from $D$. Suppose  $\perp\in A$. Since $D$ is a c.p.o, then for all $d\in D$ we have $\perp\sqsubseteq d$. Since $A$ is final, $\perp\in A$ and $\perp\sqsubseteq d$, then $d\in A$. Thus $A=D$, which is a contradiction. 
\end{proof}

\begin{teor}
\label{trivialgroupcpo-theorem}
If $D$ is a c.p.o.\ with the Scott topology, then $\pi_{n}(D,d)=\{[c^n(d)]\}$ for all $d\in D$ and $n\geq 0$. 
\end{teor}

\begin{proof}
	By Theorem \ref{trivialgroupcontrac-theorem}, it is enough to check that $D$ is contractible, i.e., one has to check that  for the identity function $I_{D}:D\rightarrow D$ there exists some constant function $f_c:D\rightarrow D$  such that $f_c\simeq I_D$.
Consider the map $H:D\times[0,1]\rightarrow D$ defined by 
\begin{equation*}
\hspace{0.5cm}H(x,t) =
\begin{cases}
\perp & \text{if $t=0$,}\\
x & \text{if $t\in (0,1]$}
\end{cases}
\end{equation*}
and  let us show that $H$ is a contraction from $D$. Clearly $H(\,\cdot\,,0)=f_\perp(\cdot)$ and $H(\,\cdot\,,t)=I_D(\cdot)$ if $t\in(0,1]$. Now take any open  $A\neq D$, then
\begin{align*}
\hspace{1cm}H^{-1}(A)&=\{(x,t)\in D\times [0,1]:H(x,t)\in A\} \\
&=\left(\{ x\in D:H(x,0)=\perp\in A\}\times\{0\}\right)\cup\left(\{x\in D:x\in A\}\times(0,1]\right) \\
&=\left(\emptyset\times\{0\}\right)\cup\left(A\times(0,1]\right) \,\,\, \text{(by Lemma \ref{bottomnotin-lemma}, $\perp\notin A$)} \\
&=A\times(0,1], 
\end{align*}
which is an open from $D\times[0,1]$, then $H$ is continuous. Thus $H$ is a contraction from $D$.
\end{proof}

Let us now generalize the notion of a continuous path from point $a$ to point $b$, to $n$-paths based on points $a$ and $b$ in any topological space $D$. 

\begin{nota}
	Take a map $p:[0,1]^n\rightarrow D$. Write $p[\,t_r]:[0,1]^{n-1}\rightarrow D$ for the map such that
	$$p[t_r](t_1,\ldots,t_{r-1},t_{r+1},\ldots,t_n):=p(t_1,\ldots,t_{r-1},t_r,t_{r+1},\ldots,t_n),$$
	for $r<s$ we denoted,
	$$p[t_r,t_s](t_1,\ldots,t_{r-1},t_{r+1},\ldots,t_{s-1},t_{s+1},\ldots,t_n):=$$
	$$p(t_1,\ldots,t_{r-1},t_r,t_{r+1},\ldots,t_{s-1},t_s,t_{s+1},\ldots,t_n),$$
	and for a fixed $a\in [0,1]$, write $p[t_r=a]:[0,1]^{n-1}\rightarrow D$ for the map such that 
	$$p[t_r=a](t_1,\ldots,t_{r-1},t_{r+1},\ldots,t_n):=p(t_1,\ldots,t_{r-1},a,t_{r+1},\ldots,t_n),$$
	for $r<s$ we wrote,
	
	$$p[t_r=a,t_s=b](t_1,\ldots,t_{r-1},t_{r+1},\ldots,t_{s-1},t_{s+1},\ldots,t_n):=$$
	$$p(t_1,\ldots,t_{r-1},a,t_{r+1},\ldots,t_{s-1},b,t_{s+1},\ldots,t_n).$$ 
\end{nota}

\begin{defin}[n-path]\label{n-path-definition}
	Let $D$ be a topological space. For each $n\in\mathbb{N}$, define an $n$-path  based in the points $a,b\in D$ as continuous function $p:[0,1]^n\rightarrow D$ such that 
	$$p[t_n=0](t_1,\ldots,t_{n-1})=a\,\,\,\text{and}\,\,\, p[t_n=1](t_1,\ldots,t_{n-1})=b,$$
	for each $t_1,\ldots,t_{n-1}\in [0,1]$.

	\begin{itemize}
		\item 	For each $r<n$, define the product $\ast_r$ of $n$-paths $p$ and $q$ as
		\begin{enumerate}
			\item if $p[t_1=1]=q[t_1=0]$, define the n-path
			\begin{equation*}
			(p\ast_{n-1} q)(t_1,\ldots,t_n):=
			\begin{cases}
			p(2t_1,t_2,\ldots,t_n) & \text{if \, $0\leq t_1\leq \frac{1}{2}$,}\\
			q(2t_1-1,t_2,\ldots,t_n) & \text{if \, $\frac{1}{2}\leq t_1\leq 1$.}
			\end{cases}
			\end{equation*}
			\item if $r<n-1$ and $p[t_1=1,\ldots,t_{(n-1)-r}=1]=q[t_1=0,\ldots,t_{(n-1)-r}=0]$, define the n-path $p\ast_r q$ such that 
			$$(p\,\ast_r\, q)[t_1,\ldots,t_{(n-1)-r}]:=p[t_1,\ldots,t_{(n-1)-r}]\ast_r\, q[t_1,\ldots,t_{(n-1)-r}].$$
		\end{enumerate}
		\item For each $n$-path $p$, define the identity $(n+1)$-path as the constant path $c(p):[0,1]^{n+1}\rightarrow D$ such that $(c(p))[t_1]=p$ for all $t_1\in [0,1]$.  
	\end{itemize}
	
\end{defin}

Define the equivalence relation $=_h$ on $n$-paths as given by: $p_1 =_h p_2$ if there exists an $(n+1)$-path $p$ from $p_1$ to $p_2$.

\medskip  Thus, the set of all equivalence classes on the set of $n$-paths along with the product of $n$-paths, it generates a groupoid; where  the product between classes is defined naturally by $[p]\ast_r [q]:=[p\ast_r q]$, which satisfies $[c^{n-r}(p[t_1=0,\ldots,t_{n-r}=0])]\ast_r [p]=[p]$ and $[p] \ast_r [c^{n-r}(p[t_1=1,\ldots,t_{n-r}=1])]=[p]$, and for .each $n$-path $p$ there is a $\bar{p}$ such that $\bar{p}[t_1,\ldots,t_{n-r}]=p[1-t_1,\ldots,1-t_{n-r}]$, for which it holds that $[p]\ast_r [\bar{p}]=[c^{n-r}(p[t_1=0,\dots,t_{n-r}=0])]$ and $[\bar{p}]\ast_r [p]=[c^{n-r}(p[t_1=1,\ldots,t_{n-r}=1])]$.

\begin{ejem}
	For $n=2$, any 2-paths $p$ and $q$ would be homotopies between paths (1-paths), whose $r$-product is given according to the cases: 
	
	\begin{enumerate}
		\item If $p[t_1=1]=q[t_1=0]$, the 1-product is 
		\begin{equation*}
		(p\ast_1 q)(t_1,t_2):=
		\begin{cases}
		p(2t_1,t_2) & \text{if \, $0\leq t_1\leq \frac{1}{2}$,}\\
		q(2t_1-1,t_2) & \text{if \, $\frac{1}{2}\leq t_1\leq 1$.}
		\end{cases}
		\end{equation*}
		so $(p\ast_1 q)[t_1=0]=p[t_1=0]$ and $(p\ast_1 q)[t_1=1]=q[t_1=1]$, i.e., $(p\ast_1 q)$ is a homotopy from path $p[t_1=0]$ to path $q[t_1=1]$ as seen in Figure \ref{fig01}. 
		
		\begin{figure}[!htb]
			\centering
			\begin{tikzcd}
			a \arrow[bend left=80]{r}[name=LUU, below]{}
			\arrow{r}[name=LUD]{}
			\arrow[swap]{r}[name=LDU]{}
			\arrow[bend right=80]{r}[name=LDD]{}
			\arrow[Rightarrow,to path=(LUU) -- (LUD)\tikztonodes]{r}{p}
			\arrow[Rightarrow,to path=(LDU) -- (LDD)\tikztonodes]{r}{q}
			& 
			b
			\end{tikzcd}
			\caption{The product $p\ast_1 q$}\label{fig01}   
		\end{figure}
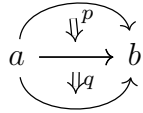

		\item If $p[t_2=1]=q[t_2=0]$, the $0$-product is given by
		\begin{equation*}
		(p\ast_0 q)(t_1,t_2)=(p\ast_0 q)[t_1](t_2)=(p[t_1]\ast_0 q[t_1])(t_2)=
		\end{equation*}
		\begin{equation*}
		=
		\begin{cases}
		p[t_1](2t_2) & \text{if \, $0\leq t_2\leq \frac{1}{2}$,}\\
		q[t_1](2t_2-1) & \text{if \, $\frac{1}{2}\leq t_2\leq 1$,}
		\end{cases}
		=
		\begin{cases}
		p(t_1,2t_2) & \text{if \, $0\leq t_2\leq \frac{1}{2}$,}\\
		q(t_1,2t_2-1) & \text{if \, $\frac{1}{2}\leq t_2\leq 1$.}
		\end{cases}
		\end{equation*}
		
		\noindent Then $(p\ast_0 q)[t_1=0]=p[t_1=0]\ast_0 q[t_1=0]$ and $(p\ast_0 q)[t_1=1]=p[t_1=1]\ast_0 q[t_1=1]$, i.e., $(p\ast_0 q)$ is a homotopy from path $p[t_1=0]\ast_0\, q[t_1=0]$ to path $p[t_1=1]\ast_0\, q[t_1=1]$ as in Figure \ref{fig02}.
		
		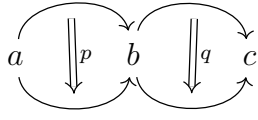
\begin{figure}[!htb]
			\begin{center}
				\begin{tikzcd}
				a \arrow[bend left=80]{r}[name=LUU, below]{}

				\arrow[bend right=80]{r}[name=LDD]{}
				\arrow[Rightarrow,to path=(LUU) --  (LDD)\tikztonodes]{r}{p}
				& 
				b
				\arrow[bend left=80]{r}[name=RUU, below]{}

				\arrow[bend right=80]{r}[name=RDD]{}
				\arrow[Rightarrow,to path=(RUU) -- (RDD)\tikztonodes]{r}{q}
				& c
				\end{tikzcd} 
				\caption{The product $p\ast_0 q$}\label{fig02}   
			\end{center}      
		\end{figure}
	\end{enumerate}
\end{ejem}

\begin{ejem}\label{Example 2-paths}
	Let the set $L=\{0,1,2\}\cup\{\bot,\top\}$. For all $a,b\in L$, define
	$$a\sqsubseteq b\Longleftrightarrow (a=\bot \,\, and \,\, b\in L) \,\, or \,\, (a\in L \,\,and \,\, b=\top) \,\, or \,\,(a=b).$$
	
	The poset $(L,\sqsubseteq)$ is a finite lattice, so its a c.p.o. which we endow with the Scott topology (see figure \ref{fig9}).  
	
	\begin{figure}[!htb]
		\centering
		\begin{tikzpicture}[scale=.7]
		\node (top) at (0,2.5) {$\top$};
		\node (0) at (-3,0) {$0$};
        \node (1) at (0,0) {$1$};
        \node (2) at (3,0) {$2$};
		\node (bot) at (0,-2.5) {$\bot$};
		\draw [-] (top) -- (0); 
		\draw[-](top) -- (1); 
		\draw[-](top) -- (2);
		\draw[-] (bot) -- (0);
		\draw[-] (1) -- (bot);
		\draw[-](2) --(bot);
		\end{tikzpicture}
		\caption{The c.p.o $L$}\label{fig9}
	\end{figure}
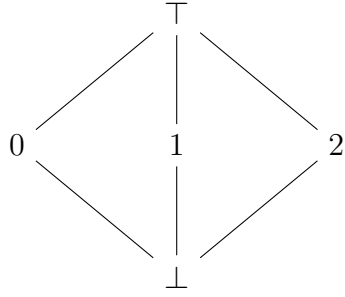 
	
	\medskip Let $a\in\{0,1,2\}$, define the 1-paths
	\begin{equation*}
	p^{a\rightarrow\top}(t):=
	\begin{cases}
	a & \text{if \, $t=0$,}\\
	\top & \text{if \, $0<t\leq 1$,}
	\end{cases}
	\hspace{1cm}
	p^{a\rightarrow\perp}(t):=
	\begin{cases}
	a & \text{if \, $0\leq t<1$,}\\
	\perp & \text{if \, $t=1$,}
	\end{cases}
	\end{equation*}
	$$p^{\top\rightarrow a}(t):=p^{a\rightarrow\top}(1-t),\hspace{1cm}p^{\perp\rightarrow a}(t):=p^{a\rightarrow\perp}(1-t).$$
	\begin{itemize}
		\item  For $a,b\in\{0,1,2\}$ we have 
		$$p^{a\rightarrow\top}\ast_0p^{\top\rightarrow b}=_h p^{a\rightarrow\perp}\ast_0p^{\perp\rightarrow b}.$$
		Since, for the paths  
		$$p_\top^{a\rightarrow b}:=p^{a\rightarrow\top}\ast_0p^{\top\rightarrow b}, \hspace{1cm}p_\perp^{a\rightarrow b}:=p^{a\rightarrow\perp}\ast_0p^{\perp\rightarrow b},$$
		there exists a 2-path $p^{a\Rightarrow b}$ such that for each $t_2\in [0,1]$
		$$p^{a\Rightarrow b}[t_1=0](t_2)=p_\top^{a\rightarrow b}(t_2)=\begin{cases}
		a & \text{if \, $t_2=0$,}\\
		\top & \text{if \, $0<t_2<1$,}\\
		b & \text{if \, $t_2=1$,}
		\end{cases}$$ 
		$$p^{a\Rightarrow b}[t_1=1](t_2)=p_\perp^{a\rightarrow b}(t_2)=\begin{cases}
		a & \text{if \, $0\leq t_2<\frac{1}{2}$,}\\
		\perp & \text{if \, $t_2=\frac{1}{2}$,}\\
		b & \text{if \, $\frac{1}{2}<t_2\leq 1$,}
		\end{cases}$$
		which is given by 
		$$p^{a\Rightarrow b}(t_1,t_2)=\begin{cases}
		\top & \text{if \, $(t_1=0$ and $0<t_2<1)$ or $(0<t_1,t_2<1)$,}\\
		a & \text{if \, $(0\leq t_1\leq1$ and $t_2=0)$ or $(t_1=1$ and $0<t_2<\frac{1}{2})$,} \\
		b & \text{if \, $(0\leq t_1\leq1$ and $t_2=1)$ or $(t_1=1$ and $\frac{1}{2}<t_2<1)$,} \\
		\perp & \text{if \, $t_1=1$ and $t_2=\frac{1}{2}.$}
		\end{cases}$$
		On the other hand, if we defined the path
		$$e(t):=\begin{cases}
		\top & \text{if \, $0\leq t<1$,}\\
		\perp & \text{if \, $t=1$.}\\
		\end{cases}$$   
		
		The 0-product of the 2-paths $p^{0\Rightarrow 1}$ and $p^{1\Rightarrow 2}$ is given by
		\begin{align*}
		(p^{0\Rightarrow 1}\ast_0 p^{1\Rightarrow 2})[t_1]&=p^{0\Rightarrow 1}[t_1]\ast_0 p^{1\Rightarrow 2}[t_1] \\
		&=p_{e(t_1)}^{0\rightarrow 1}\ast_0 p_{e(t_1)}^{1\rightarrow 2} \\
		&=(p^{0\rightarrow e(t_1)}\ast_0p^{e(t_1)\rightarrow 1})\ast_0(p^{1\rightarrow e(t_1)}\ast_0p^{e(t_1)\rightarrow 2}) \\
		&=_hp^{0\rightarrow e(t_1)}\ast_0(p^{e(t_1)\rightarrow 1}\ast_0p^{1\rightarrow e(t_1)})\ast_0p^{e(t_1)\rightarrow 2} \\
		&=_hp^{0\rightarrow e(t_1)}\ast_0 1_{e(t_1)}\ast_0p^{e(t_1)\rightarrow 2} \\
		&=_hp^{0\rightarrow e(t_1)}\ast_0p^{e(t_1)\rightarrow 2} \\
		&=p_{e(t_1)}^{0\rightarrow 2} \\
		&=p^{0\Rightarrow 2}[t_1].
		\end{align*}
		Therefore,
		$$p^{0\Rightarrow 1}\ast_0 p^{1\Rightarrow 2}=_hp^{0\Rightarrow 2}.$$
		
		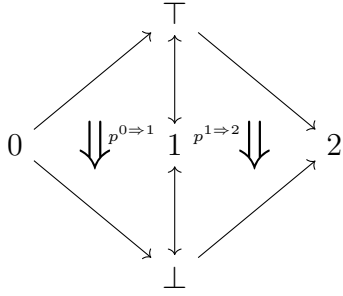
\begin{figure}[!htb]
			\centering
			\begin{tikzpicture}[scale=.7]
			\node (top) at (0,2.5) {$\top$};
			\node (0) at (-3,0) {$0$};
			\node (p1) at (-0.8,0.2) {{\tiny $p^{0\Rightarrow 1}$}};
			\node (f1) at (-1.5,0) {{\LARGE $\Downarrow$}};
			\node (1) at (0,0) {$1$};
			\node (p2) at (0.8,0.2) {{\tiny $p^{1\Rightarrow 2}$}};
			\node (f2) at (1.5,0) {{\LARGE $\Downarrow$}};
			\node (2) at (3,0) {$2$};
			\node (bot) at (0,-2.5) {$\bot$};
			\draw [<-] (top) -- (0); 
			\draw[<->](top) -- (1); 
			\draw[->](top) -- (2);
			\draw[<-] (bot) -- (0);
			\draw[<->] (1) -- (bot);
			\draw[<-](2) --(bot);
			\end{tikzpicture}
			\caption{The product $p^{0\Rightarrow 1}\ast_0 p^{1\Rightarrow 2}$}\label{fig7}
		\end{figure}

		\item  Let $a\in\{0,1,2\}$, define the path
		$$q_a(t):=(p^{\top\rightarrow a}\ast_0p^{a\rightarrow \perp})(t)=\begin{cases}
		\top & \text{if \, $0\leq t<\frac{1}{2}$,}\\
		a & \text{if \, $\frac{1}{2}\leq t<1$,}\\
		\perp & \text{if \, $t=1$.}
		\end{cases}$$
		
		For $a,b\in\{0,1,2\}$ we have 
		$$q_a=p^{\top\rightarrow a}\ast_0p^{a\rightarrow \perp}=_hp^{\top\rightarrow b}\ast_0p^{b\rightarrow \perp}=q_b.$$
		Since, there is a 2-path $q^{a\Rightarrow b}$ such that 
		$$q^{a\Rightarrow b}[t_1=0]=q_a,\hspace{1cm}q^{a\Rightarrow b}[t_1=1]=q_b,$$
		which is given by
		$$q^{a\Rightarrow b}(t_1,t_2)=\begin{cases}
		\top & \text{if \, $(0\leq t_1\leq 1$ and $0\leq t_2<\frac{1}{2})$ or $(0<t_1<1$ and $t_2=\frac{1}{2})$}\\
		& \hspace{0.6cm} \text{or \,$(0<t_1<1$ and $\frac{1}{2}<t_2<1),$ } \\
		a & \text{if \, $t_1=0$ and $\frac{1}{2}\leq t_2<1,$} \\
		b & \text{if \, $t_1=1$ and $\frac{1}{2}\leq t_2<1,$} \\
		\perp & \text{if \, $0\leq t_1\leq1$ and $t_2=1.$}
		\end{cases}$$
		
		If $0\leq t_2<\frac{1}{2}$ or $t_2=1$, the 1-product of the 2-paths $q^{0\Rightarrow 1}$ and $q^{1\Rightarrow 2}$ is given by 
		$$ (q^{0\Rightarrow 1}\ast_1 q^{1\Rightarrow 2})[t_2]=1_{e(t_2)}=q^{0\Rightarrow 2}[t_2],$$
		\medskip and if $\frac{1}{2}\leq t_2<1$, the 1-product results in
		\begin{align*}
		(q^{0\Rightarrow 1}\ast_1 q^{1\Rightarrow 2})[t_2]&=q^{0\Rightarrow 1}[t_2]\ast_0 q^{1\Rightarrow 2}[t_2] \\
		&=p_{\top}^{0\rightarrow 1}\ast_0 p_{\top}^{1\rightarrow 2} \\
		&=_hp_{\top}^{0\rightarrow 2} \\
		&=q^{0\Rightarrow 2}[t_2].
		\end{align*}
		Thus,
		$$q^{0\Rightarrow 1}\ast_1 q^{1\Rightarrow 2}=_hq^{0\Rightarrow 2}.$$
			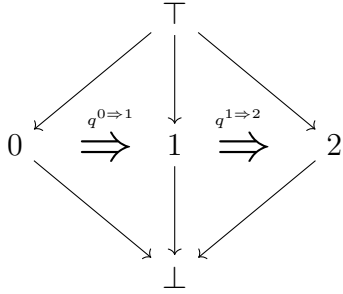
\begin{figure}[!htb]
			\centering
			\begin{tikzpicture}[scale=.7]
			\node (top) at (0,2.5) {$\top$};
			\node (0) at (-3,0) {$0$};
			\node (q1) at (-1.2,0.5) {{\tiny $q^{0\Rightarrow 1}$}};
			\node (f1) at (-1.3,-0.1) {{\LARGE $\Rightarrow$}};
		    \node (q2) at (1.2,0.5) {{\tiny $q^{1\Rightarrow 2}$}};
            \node (f1) at (1.3,-0.1) {{\LARGE $\Rightarrow$}};
            \node (1) at (0,0) {$1$};
			\node (2) at (3,0) {$2$};
			\node (bot) at (0,-2.5) {$\bot$};
			\draw [->] (top) -- (0); 
			\draw[->](top) -- (1); 
			\draw[->](top) -- (2);
			\draw[<-] (bot) -- (0);
			\draw[->] (1) -- (bot);
			\draw[->](2) --(bot);
			\end{tikzpicture}
			\caption{The product $q^{0\Rightarrow 1}\ast_1 q^{1\Rightarrow 2}$}\label{fig8}
		\end{figure}
	\end{itemize}
\end{ejem}

\begin{defin}[Parallel paths]
	Two $n$-paths $p$ and $q$ based at, resp., $a,b\in D$ are parallel if for each $r=1\ldots,n-1$ one has that
	
	\begin{enumerate}
		\item $p[t_1=0,\ldots,t_r=0]=q[t_1=0,\ldots,t_r=0]$,
		\item $p[t_1=1,\ldots,t_r=1]=q[t_1=1,\ldots,t_r=1].$
	\end{enumerate} 
\end{defin}

\begin{corol}\label{equal-paths-corollary}
	If $p$ and $q$ are parallel $n$-paths based on $a,b\in D_\infty$, then $p=_h q$.   
\end{corol}
\begin{proof}
	Since $p$ and $q$ are parallel, we have $$p[t_1=1,\ldots t_{n-1}=1]=q[t_1=1,\ldots,t_{n-1}=1]=\bar{q}[t_1=0,\ldots,t_{n-1}=0],$$ by definition of product
	$$(p\ast_0\bar{q})[t_1,\dots,t_{n-1}]= p[t_1,\dots,t_{n-1}]\ast_0 \bar{q}[t_1,\ldots,t_{n-1}]$$
	where $$(p\ast_0\bar{q})[t_1,\dots,t_{n-1}](0)=p[t_1,\dots,t_{n-1}](0)=a; \,\,\, p(t_1,\ldots,t_{n-1},0)=a,$$
	$$(p\ast_0\bar{q})[t_1,\dots,t_{n-1}](1)=\bar{q}[t_1,\dots,t_{n-1}](1)=a; \,\,\, q(t_1,\ldots,t_{n-1},1)=b,$$
	for each $t_1,\dots,t_{n-1}\in [0,1]$. 
	
	\medskip Then   $(p\ast_0\bar{q})[t_1,\dots,t_{n-1}]$ is a 1-path closed in $a$ for all $t_1,\dots,t_{n-1}\in [0,1]$. By Theorem \ref{trivialgroupcpo-theorem}, one has $(p\ast_0 \bar{q})[t_1,\ldots,t_{n-1}]=_hc(a)=c^n(a)[t_1,\ldots,t_{n-1}]$ for each $t_1,\ldots,t_{n-1}\in[0,1]$, where $c(a)$ is the constant path in $a$, i.e., $c(a)(t)=a$ for each $t\in[0,1]$, and $c^n(a)[t_1,\ldots,t_r]=c^{n-r}(a)$. Thus $p\ast_0 \bar{q}=_h c^n(a)$, so $p=_hq$. 
\end{proof}

\begin{nota}
	Write $\Pi_{n}(D_\infty,a,b)$ for the  (weak) $n$-groupoid of $n$-paths based at $a,b\in D_\infty$. And write $\Pi_\infty(D_\infty,a,b)$ for the globular set 
	$$\cdots\rightrightarrows^s_t \Pi_n(D_\infty,a,b)\rightrightarrows^s_t \Pi_{n-1}(D_\infty,a,b)\rightrightarrows^s_t\cdots\rightrightarrows^s_t \Pi_1(D_\infty,a,b)\rightrightarrows^s_t \Pi_0(D_\infty,a,b),$$
	where $s(p):=p[t_1=0]$ and $t(p):=p[t_1=1]$. 	
\end{nota}

\medskip Therefore, by Corollary \ref{equal-paths-corollary}, the $n$-groupoid $\Pi_{n}(D_\infty,a,b)$ is trivial for each $n\geq 0$, i.e., under the intensional equality $=_h$  there is only one parallel $n$-path which inhabits $\Pi_{n}(D_\infty,a,b)$. Since $D_\infty$ is connected by paths (by Theorem \ref{trivialgroupcpo-theorem}, $\pi_0(D_\infty,d)\cong \{ 0 \} $ for all $d\in D_\infty$), given $a',b'\in D_\infty$ it holds that $\Pi_{n}(D_\infty,a,b)\cong\Pi_{n}(D_\infty,a',b')$. Thus any $n$-groupoid $\Pi_{n}(D_\infty,a,b)$ can be written simply as $\Pi_{n}(D_\infty)$.

\medskip On the other hand, notice that if $a=b$ then the $n$-groupoid $\Pi_n(D_\infty,a,b)=\Pi_n(D_\infty,a,a)=\pi_n(D_\infty,a)$ for each $n\geq 0$. So  $\Pi_\infty(D_\infty,a,a)=\pi_\infty(D_\infty,a)$, where the $\infty$-group $\pi_\infty(D_\infty,a)$ is $\infty$-globular set which corresponds to diagram 
  $$\cdots\rightrightarrows^s_t \pi_n(D_\infty,a)\rightrightarrows^s_t \pi_{n-1}(D_\infty,a)\rightrightarrows^s_t\cdots\rightrightarrows^s_t \pi_1(D_\infty,a)\rightrightarrows^s_t \pi_0(D_\infty,a),$$
  with $s=t$.
  
\section{The $\lambda$-model $\mathfrak{D}_\infty$ and its fundamental $\infty$-groupoid}

\medskip According to Definition \ref{Setinftygrupoid_D-definition}, let $\mathfrak{D}_{\infty}$  be the $\infty$-groupoid generated by the c.p.o.\ $D_\infty$ with the Scott topology. By Proposition \ref{prop1} we have that for each $\mathfrak{d}\in\mathfrak{D}_\infty$  there is $d\in D_\infty$ such that
\begin{align*}
\hspace{2cm}\mathfrak{d}&=\langle\pi_0(D_\infty,d),\pi_1(D_\infty,d),\pi_2(D_\infty,d),\ldots\rangle \\
&=\langle \{d\},\{[c(d)]\},\{[c^2(d)]\},\ldots\rangle \\
&\cong\langle [d],[c(d)],[c^2(d)],\ldots,\rangle\in\pi_\infty(D_\infty,d) \,\,\, \text{(by Definition \ref{Setinftygrupoid_D-definition}}). 
\end{align*} 

	Therefore $\mathfrak{D}_\infty(d)\cong\pi_\infty (D_\infty,d)$ (isomorphism of groups) and each $\mathfrak{d}\in \mathfrak{D}_\infty$ can be seen as the infinity matrix 

$$ \mathfrak{d}\cong\langle d,c_d, c_{c_d},\dots\rangle:=
\left(\begin{array}{cccc}
d_0 & c_{d_0} & c_{c_{d_0}} &\cdots  \\ 
d_1 & c_{d_1} & c_{c_{d_1}} &\cdots  \\
d_2 & c_{d_2} & c_{c_{d_2}} &\cdots  \\
\vdots & \vdots & \vdots & \ddots \\
\end{array} \right) $$	

\begin{defin}[Application in $\mathfrak{D}_\infty$]\label{testenv-definition}
	For $\mathfrak{a}, \mathfrak{b}\in\mathfrak{D}_\infty$ such that $\mathfrak{a}_n=\pi_{n}(D_\infty,a)$ and $\mathfrak{b}_n=\pi_{n}(D_\infty,b)$, define the  product $\pi_{n}(D_\infty,a)$ with $\pi_{n}(D_\infty,b)$ as
	$$\pi_{n}(D_\infty,a)\bullet\pi_{n}(D_\infty,b):=\pi_{n}(D_\infty,a\bullet b),$$
	where $a\bullet b$ was defined below Definition \ref{D_infty-definition},  so the application of $\mathfrak{a}$ to $\mathfrak{b}$ in $\mathfrak{D}_\infty$ as the infinite sequence
	$$\mathfrak{a}\bullet\mathfrak{b}=\left\langle \mathfrak{a}_0\bullet \mathfrak{b}_0,\mathfrak{a}_1\bullet \mathfrak{b}_1,\mathfrak{a}_2\bullet \mathfrak{b}_2,\ldots \right\rangle=\left\langle a\bullet b,\pi_1(D_\infty,a\bullet b),\pi_2(D_\infty,a\bullet b),\ldots \right\rangle.$$
\end{defin}

\begin{teor}
\label{testenv-theorem}
$\langle\mathfrak{D}_\infty,\bullet\rangle\cong\langle D_\infty,\bullet\rangle$. So $\langle\mathfrak{D}_\infty,\bullet\rangle$ is an extensional $\lambda$-model.
\end{teor}

\begin{proof}
	It is enough to show that the mapping $F: \langle D_\infty,\bullet\rangle\rightarrow\langle\mathfrak{D}_\infty,\bullet\rangle$ such that $(F(a))_n=\pi_n(D_\infty,a)$ for each $n\in\mathbb{N}$, is an  isomorphism. 
\begin{align*}
\hspace{0.5cm}(F(a)\bullet F(b))_n&=(F(a))_n\bullet (F(b))_n \\
&=\pi_n(D_\infty,a)\bullet\pi_n(D_\infty,b) \\
& =\pi_n(D_\infty,a\bullet b) \\
& =\left(F(a\bullet b) \right)_n,
\end{align*}
for all $n\in\mathbb{N}$. This is $F(a)\bullet F(b)=F(a\bullet b)$.

\medskip $F$ is injective. Let $F(a)=F(b)$, by Proposition \ref{prop37} we have 
$$a=(F(a))_0=(F(b))_0=b.$$

$F$ is surjective. Let $\mathfrak{a}\in\mathfrak{D}_\infty$, then we have  $\mathfrak{a}_0=a$ for each $a\in D_\infty$. Thus $F(a)=\mathfrak{a}$.
\end{proof}

\begin{defin}[Partial order in $\mathfrak{D}_\infty$]
	For each $\mathfrak{a}$ and $\mathfrak{b}$ in $\mathfrak{D}_\infty$ define the partial order on $\mathfrak{D}_\infty$ as 
	$$\mathfrak{a}\sqsubseteq\mathfrak{b}\Longleftrightarrow a\sqsubseteq b,$$
	where $\mathfrak{a}_n=\pi_{n}(D_\infty,a)$ and $\mathfrak{b}_n=\pi_{n}(D_\infty,b)$. 
\end{defin}

\begin{rem} $\langle\mathfrak{D}_\infty,\sqsubseteq\rangle\cong\langle D_\infty,\sqsubseteq\rangle$. So $\langle\mathfrak{D}_\infty,\sqsubseteq\rangle$ is a c.p.o.\ and the induced topology by the mapping $F:\langle D_\infty,\tau_{Scott}\rangle\rightarrow\mathfrak{D}_\infty$ of the Theorem \ref{testenv-theorem}'s proof to the set $\mathfrak{D}_\infty$  is exactly the Scott topology.  
\end{rem}

\section{Interpretation of $\beta$-equality proofs in $D_\infty$}

In $\lambda$-calculus we have that two $\lambda$-terms $M$ and $N$ are $\beta$-equal, $M=_\beta N$, if there is a sequence of $\lambda$-terms $N_1$, $N_2$,\ldots,$N_n$ such that
$$(\forall i\leq n-1)(N_i\vartriangleright_{1\beta} N_{i+1} \,\,\, or \,\,\, N_{i+1}\vartriangleright_{1\beta} N_i \,\,\, or \,\,\, N_i\equiv_\alpha N_{i+1}),$$
where $N_1=M$ and $N_n=N$. Thus the equality of the theory $\lambda\beta$ can be seen as an intensional equality, in the sense that the chain   
$$M=N_0=_{\beta}N_1=_{\beta}\cdots=_{\beta}N_n=N,$$   
would be a proof $P$ of equality $M=_\beta N$, which can be interpreted in some topological model $\langle D,\bullet, \llbracket \, \rrbracket\rangle$ as a continuous path $p:\llbracket M \rrbracket\rightsquigarrow \llbracket N \rrbracket$
which passes through the intermediate points $\llbracket N_1 \rrbracket,\llbracket N_2 \rrbracket\ldots,\llbracket N_{n-1} \rrbracket$. Then we could ask ourselves if given two 1-proofs $P$ and $Q$ of equality $M=_\beta N$, in space $D$, is there a homotopy (2-path) between the paths $p:=\llbracket P \rrbracket$ and $q:=\llbracket Q \rrbracket$?. Now if we have some intensional definition (with respect to $D$) of $D$-equality between the proofs of equality $P$ and $Q$ such that its interpretation int $D$ is a homotopy from $p$ to $q$, we could ask again if for the 2-proofs $F$ and $G$ of the equality $P=_D Q$ is there a homotopy of homotopies (3-path) between the homotopies $f:=\llbracket F \rrbracket$ and $g:=\llbracket G \rrbracket$? And so on, we can continue asking with the purpose of forming from model topology $D$ an $\infty$-groupoid structure in $\lambda$-calculus.   

\medskip We have that $D_\infty$ is a topological $\lambda$-model, but its topological structure does not allow to capture relevant information about equality between higher proofs at $\lambda$-calculus, since the $\infty$-groupoid generated by $D_\infty$ is trivial. The reason is that any proof $P$ of equality $M=_\beta N$ given by the chain
$$P:M=N_0=_{\beta}N_1=_{\beta}\cdots=_{\beta}N_n=N,$$
would be interpreted by some path $p:\llbracket M \rrbracket\rightsquigarrow\llbracket N \rrbracket $ that passes through the intermediate points $\llbracket N_1 \rrbracket,\ldots,\llbracket N_{n-1} \rrbracket$, but all these points are equal in space $D_\infty$, i.e.,
$$p:\llbracket M\rrbracket=\llbracket N_0\rrbracket=\llbracket N_1\rrbracket=\cdots=\llbracket N_n\rrbracket=\llbracket N\rrbracket.$$

Therefore the interpretation of proof $P$ is some closed path $p:\llbracket M \rrbracket\rightsquigarrow \llbracket M \rrbracket$, we wrote such interpretation as $\llbracket P \rrbracket:=p$. Since $\pi_1(D_\infty,\llbracket M \rrbracket)$ is trivial by Theorem \ref{trivialgroupcpo-theorem}, then $p$ is homotopically equal to the constant path $c(\llbracket M \rrbracket)$, i.e., $p=_hc(\llbracket M \rrbracket)$. Now given any other proof 
$$Q:M=N_0'=_{\beta}N_1'=_{\beta}\cdots=_{\beta}N_n'=N,$$
of equality $M=_\beta N$, with interpretation $\llbracket Q \rrbracket=q$, by Corollary \ref{equal-paths-corollary} we have $p=_hq$, so we could assert that  $P=_{D_\infty}Q$ (see Figure \ref{fig:inter_Dinf}). Thus the class of all the proofs of any equality is trivial with respect to $D_\infty$.

\begin{figure}[ht!]
	\centering
    \fcolorbox{gray}{white}{\includegraphics[width=0.50\textwidth]{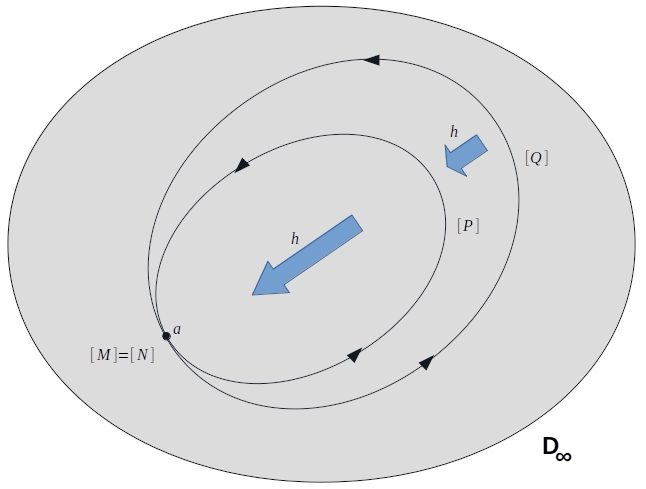}}
	\caption{Interpretation of equal proofs $P,Q:(M=_\beta N)$ on $D_\infty$.}
	\label{fig:inter_Dinf}
\end{figure}

\medskip If we continue at the next level, i.e., given any 2-proofs $F$ and $G$ of equality $P=_{D_\infty}Q$, we have that the interpretation of $F$ and $G$ is given for some pair of homotopies (2-paths) $f,g:p \rightsquigarrow q$, where $p,q:\llbracket M \rrbracket\rightsquigarrow\llbracket M \rrbracket$, by Corollary \ref{equal-paths-corollary} it has $f=_h g$ thus $F=_{D_\infty}G$. Therefore the class of all 2-proofs of any proof equality $P=_{D_\infty} Q$ is also trivial.

\medskip A better way to study the intentionality of equality $=_{\beta}$ would be to set aside the  set equality of the extensional model definition  and opt rather for homotopic models (or some homotopy variation) defined below.

\begin{defin}[Homotopic $\lambda$-model]
	A  homotopic $\lambda$-model is a triple $\langle D,\bullet, \llbracket \, \rrbracket\rangle$, where $D$ is a topological space, $\bullet:D\times D\rightarrow D$ is a binary operation and $\llbracket \, \rrbracket$ is a mapping which assigns to $\lambda$-term $M$ and each assignment $\rho:Var\rightarrow D$ an element $\llbracket M \rrbracket_{\rho}$ of $D$ such that
	\begin{enumerate}
		\item $\llbracket x \rrbracket=\rho(x);$
		\item $\llbracket PQ \rrbracket_{\rho}=_h\llbracket P\rrbracket_{\rho}\bullet\llbracket Q\rrbracket_{\rho};$
		\item $\llbracket\lambda x.P\rrbracket_{\rho}\bullet d=_{h}\llbracket P\rrbracket_{[d/x]\rho}$ for all $d\in D$;
		\item $\llbracket M\rrbracket_{\rho}=\llbracket M\rrbracket_{\sigma}$ if $\rho(x)=\sigma(x)$ for $x\in FV(M)$;
		\item $\llbracket\lambda x.M\rrbracket_{\rho}=_h\llbracket\lambda y.[y/x]M\rrbracket_{\rho}$ if $y\notin FV(M)$;
		\item if $(\forall d\in D)\left(\llbracket P\rrbracket_{[d/x]\rho}=_h\llbracket Q\rrbracket_{[d/x]\rho}\right) $, then $\llbracket\lambda x.P\rrbracket_{\rho}=_h\llbracket\lambda x.Q\rrbracket_{\rho}$.
	\end{enumerate}
	The homotopic model $\langle D,\bullet, \llbracket \, \rrbracket\rangle$ is an extensional homotopic model if it satisfies the additional property: $\llbracket\lambda x.Mx\rrbracket_{\rho}=_h\llbracket M\rrbracket_{\rho}$ with $x\notin FV(M)$.
\end{defin}

To solve the triviality problem of proofs interpretation on $D_\infty$, we would have to propose a $\lambda$-homotopic model $D$ with another topology, for which there must exist two proofs $P,Q:(M =_\beta N)$, whose interpretations are not homotopically equal, $p\neq_h q$ (of course there must also be different equality proofs whose interpretations are homotopically equal), as can be seen in the Figure \ref{fig:inter_D}. It would allow us to capture more information about the multiple $\beta$-contractions and reverse $\beta$-contractions of an equality proof than a traditional model based on extensional equality between sets.

\begin{figure}[ht!]
	\centering
    \label{fig:inter_D}
	\fcolorbox{gray}{white}{\includegraphics[width=0.60\textwidth]{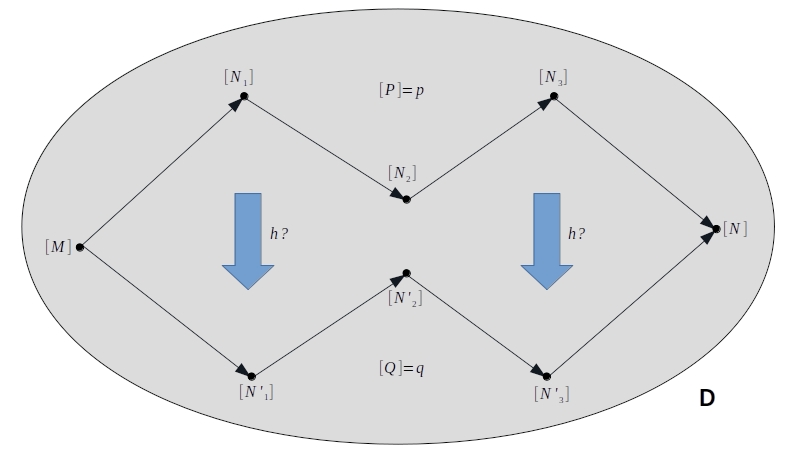}}
	\caption{Proofs $P,Q:(M=_\beta N)$ on an homotopic $\lambda$-model $D$.}
\end{figure}

\medskip In (Martínez; de Queiroz, 2020) it is put forward a cartesian closed category of $\infty$-groupoids (Kan complexes) with enough points, appropriate  for the construction of concrete \textit{homotopic $\lambda$-models} with a non-trivial structure of $\infty$-groupoid.

\medskip On other hand, if we forget the homotopies between continuous paths in $D_\infty$ (or $\mathfrak{D}_\infty$) and consider simply extensional equality between functions, we could define the interpretation of the equality proof $P:M=N_0=_{\beta}\cdots=_{\beta}N_n=N$ as a concatenation of continuous paths 
$$p:=r_1\ast r_2\ast r_3\ast\cdots\ast r_n,$$
where each continuous path $r_i:[0,1]\rightarrow D_\infty$ is given by
\begin{equation*}
	r_i(t) :=
	\begin{cases}
		a & \text{if $t\in[0,1/2)\cup(1/2,1]$,}\\
		\bot & \text{if $t=1/2$}
	\end{cases}
\end{equation*}
with $a=\llbracket M \rrbracket=\llbracket N_1 \rrbracket=\cdots=\llbracket N_n \rrbracket=\llbracket N \rrbracket$. We write the interpretation of $P$ in $D_\infty$ as $\llbracket P \rrbracket:=p$ and  each $r_i$ is called a \textit{time period}, thus we say that $p$ consists of $n$ \textit{time periods} and is written as $t(p)=n$.

\medskip Now if we have another proof of equality $	Q:M=N'_0=_{1\beta}\cdots=_{1\beta}N'_m=N$ whose interpretation would be
$$q:=r_1\ast r_2\ast r_3\ast\cdots\ast r_m,$$
where it is clear that $t(q)=m$. So we say that the equality proofs $P,Q:M=_{\beta}N$ are equal according to model $D_\infty$, noted by  if their respective interpretations are equal in the traditional sense of set theory, i.e.,
$$P=_{D_\infty}Q\Longleftrightarrow p=q,$$
thus we would have  
$$P=_{D_\infty}Q\Longleftrightarrow t(p)=t(q).$$

Thus, we have that two proofs $P$ and $Q$ of $M=_{\beta}N$  are ``equal'' if they require the same time period to complete the proof or else if  $P$ and $Q$ are sequences of the same length.

\medskip Although the extensional equality $p=q$  manages to capture information about the length of the $P$ and $Q$ proofs in $\lambda$-calculus, the nature of its extensionality  does not allow to capture more information about the 2-proofs of proof equality: $P =_{D_\infty}Q$, since there is only one canonical way to prove $P=_{D_\infty}Q$, so the generated $\infty$-groupoid by $=_{D_\infty}$ it would be trivial for 2-equality, 3-equality and so on.

\bigskip Another alternative to deal with equality of paths is offered by an approach to propositional equality which considers proofs of equality  formalized as sequences of rewrites between terms of lambda-calculus (de Queiroz; de Oliveira; Ramos 2016). Each definitional equality ($\beta$, $\eta$, $\xi$, $\mu$, $reflexivity$, $symmetry$, $transitivity$) is associated with a constant identifier, and paths are characterised as compositions of those primitive rewrites. Thus, by considering as sequences of rewrites and substitution, it comes a rather natural fact that two (or more) distinct proofs may be yet canonical and are none to be preferred over one another. By looking at proofs of equality as rewriting (or computational) paths this approach will be in line with the recently proposed connections between type theory and homotopy theory via identity types, since elements of identity types will be, concretely, paths (or homotopies). 

As a matter of fact, this is part of an ongoing project (Ramos; de Queiroz; de Oliveira, 2017, Veras et al., 2019a, Veras et al., 2019b, Veras et al., 2020), while it looks for the use of homotopy structures such as groupoids in the study of semantics of computation, it also seeks to demonstrate the utility and the impact of the so-called Curry--Howard interpretation of logical deduction in the actual practice of an important area of mathematics, namely homotopy theory. The short citation for the Royal Swedish Academy of Sciences' ``2020 Rolf Schock Prize in logic and philosophy" says that it was awarded to Per Martin-L\"of (shared with Dag Prawitz) ``for the creation of constructive type theory." In a longer statement, the prize committee recalls that constructive type theory is ``a formal language in which it is possible to express constructive mathematics" (...) ``[which] also functions as a powerful programming language and has had an enormous impact in logic, computer science and, recently, mathematics." 

In fact, by introducing a framework whose formalization of the logical notion of equality is done via the so-called ``identity type", one has the possibility for a surprising connection between term rewriting and geometric concepts such as path and homotopy. And indeed, Martin-L\"of's type theory (MLTT) allows for making useful bridges between theory of computation, algebraic topology, logic, categories, and higher algebra, and a single concept seems to serve as a bridging bond: ``path". Its impact in mathematics has been felt more strongly since the start of Vladimir Voevodsky's program on the univalent foundations of mathematics around 2005, and one specific aspect which we would like to address here is the calculation of fundamental groups of surfaces. Taking from the Wikipedia entry on ``homotopy group", calculation of homotopy groups is in general much more difficult than some of the other homotopy invariants learned in algebraic topology. Now, by using an alternative formulation of the ``identity type" which provides an explicit formal account of ``path", operationally understood as an invertible sequence of rewrites (such as Church's ``conversion"), and interpreted as a homotopy, we have provided examples of calculations of fundamental groups of surfaces such as the circle, the torus, the 2-holed torus, the Klein bottle, and the real projective plane. We would like to suggest that these examples might bear witness to the impact of MLTT in mathematics by offering formal tools to calculate and prove fundamental groups, as well as allowing to make such calculations and proofs amenable to be dealt with by systems of formal mathematics and interactive theorem provers such as Coq, Lean, and similar ones.

\section{Conclusions}

Starting from any topological space that models extensional $\lambda$-calculus, we have proposed a method to build an $\infty$-groupoid. This construction was applied to the particular c.p.o.\ $D_\infty$ with Scott topology, resulting in a constant cell infinite sequences set, where each cell sequence is isomorphic to a constant higher paths infinite matrix.

\medskip Going further, a natural way forward is to try to build  homotopic $\lambda$-model in order to avoid trivialities in the fundamental $\infty$-groupoid associated with the topology of the model, which would allow to capture relevant information about the higher equality proof in the  $\lambda$-calculus syntax.

\section*{References}

\begin{enumerate}
	\item L. Acosta and M. Rubio, Topolog\'{\i}a de Scott para relaciones de preorden, \textit{Bolet\'{\i}n de Matem\'aticas}, Nueva Serie IX No. 1 (2002), 1-10.
	\item S. Awodey and M. Warren, Homotopy theoretic models of identity types, \textit{Mathematical Proceedings of the Cambridge Philosophical Society}, v. 146, n. 1, p. 1--10, 2009.
	\item H. P. Barendregt, \textit{The Lambda Calculus, its Syntax and Semantics}, North-Holland Co., Amsternam, 1984.
	\item B. Berg and R. Garner, Types are weak $\omega$-groupoids, \textit{Proceeding of the London Mathematical Society}, v. 102, n. 2, p. 370-394, 2011. 
    \item J. M. Greenberg, \textit{Lectures on Algebraic Topology}, W. A. Benjamin, 1967.
    \item P. Goerss and J. Jardine, \textit{Simplicial Homotopy Theory}, Birkhäuser Basel, Springer Nature Switzerland AG, 2009.  
	\item A. Hatcher, \textit{Algebraic Topology}, Cambridge University Press, New York, NY, 2001.
	\item J.R. Hindley and J.P. Seldin, \textit{Lambda-Calculus and Combinators, an Introduction}, Cambridge University Press, New York, NY, 2008.
	\item M. Hofmann and T. Streicher, The groupoid model refutes uniqueness of identity proof, \textit{Logic and Computer Science}, p. 208-212, 1994. 
	\item C. Kapulkin, P. Lumsdaine, V. Voevodsky, The simplicial model of univalent foundations, \textit{ arXiv:1211.2851}, 2012.
	\item T. Leinster, \textit{Higher Operads, Higher Categories}, \textit{arXiv:math/0305049}, 2003.
	\item P. Lumsdaine, M. Shulman, Semantics of higher inductive types, \textit{Mathematical Proceedings of the Cambridge Philosophical Society} 169 (2020)
	159–208.
	\item P. Martin-L\"of, An intuitionistic theory of types: predicative part, in: \textit{Logic Colloquium '73} (Bristol, 1973).
	\item J. P. May, \textit{A Concise Course in Algebraic Topology}. University of Chicago Press, 1999. 
	\item D. Mart\'inez and R. de Queiroz, Towards a Homotopy Domain Theory, \textit{arXiv:2007.15082}, 2020.
     \item T. U. F. Program, \textit{Homotopy Type Theory: Univalent Foundations of Mathematics}. Princeton, Institute for Advanced Study, 2013. 
     \item R. J. G. B. de Queiroz, A. G. de Oliveira and A. F. Ramos. \textit{South American Journal of Logic}
Vol. 2, n. 2, pp. 245--296, 2016. (Preliminary version \textit{arXiv:1107.1901}, 2011)
\item  A. F. Ramos,  R. J. G. B.  de Queiroz, A. G.  de Oliveira. On the identity type as the type of computational paths. \textit{Logic Journal of the IGPL}, Volume 25, Issue 4, pp.\ 562--584, August 2017.
	\item D.S. Scott, A type-theoretical alternative to ISWIM, CUCH, OWHY,  \textit{Theoretical Computer Science} 121:411--440, 1993. (Informally circulated in 1969).
	\item T. M.L. Veras, A. F. Ramos,  R. J. G. B.  de Queiroz, A. G.  de Oliveira. An alternative approach to the calculation of fundamental groups based on labeled natural deduction. \textit{arXiv:1906.09107}, 2019a.
	\item T. M.L. Veras, A. F. Ramos,  R. J. G. B.  de Queiroz, A. G.  de Oliveira. A Topological Application of Labelled Natural Deduction. \textit{arXiv:1906.09105}, 2019b.
	\item T. M.L. Veras, A. F. Ramos,  R. J. G. B.  de Queiroz, T. D. O. Silva, A. G.  de Oliveira. Computational Paths -- A Weak Groupoid. \textit{arXiv:2007.07769}, 2020.
    \item V. Voevodsky, The equivalence axiom and univalent models of type theory, \textit{arXiv:1402.5556}, 2010.
\end{enumerate}

\end{document}